\documentclass[11pt,a4paper]{article}

\usepackage[T1]{fontenc}
\usepackage[utf8]{inputenc}
\usepackage{graphicx} % Required for inserting images
\usepackage{amsmath}
\usepackage{amssymb}
\usepackage{mathtools}
\usepackage{amsthm}
\mathtoolsset{showonlyrefs=true}  % show only equation numbers that are used.
\usepackage{xcolor} % EJOR asks that changes be marked (e.g. in red) in revisions
\usepackage[top=2.5cm, bottom=2.5cm, left=2.5cm, right=2.5cm]{geometry}
\usepackage{setspace}
\usepackage{caption}
\usepackage[round,authoryear]{natbib}
\bibpunct{(}{)}{;}{a}{,}{,}

\usepackage{comment}

\newcommand{\ra}{\alpha}
\newcommand{\rta}{\tilde{\alpha}}
\newcommand{\rN}{\mathcal{N}(G)}
\newcommand{\rA}[1]{\mathrm{#1}}

\newtheorem{theorem}{Theorem}
\newtheorem{lemma}[theorem]{Lemma}
\newtheorem{proposition}[theorem]{Proposition}
\newtheorem{corollary}[theorem]{Corollary}

\begin{document}

% ------------------------------------------------------------
% TITLE PAGE (EJOR: title, all author names, affiliations with
% full postal addresses, corresponding author with e-mail).
% >>> REPLACE THE PLACEHOLDERS BELOW WITH THE ACTUAL DETAILS <<<
% ------------------------------------------------------------
\begin{center}
{\LARGE\bfseries Refundable Deposits: How to Restore Cooperation in Finitely Repeated Games\par}

\vspace{1.2em}

{\large
Giulio Salizzoni\textsuperscript{a}%
\quad Domenico Mergoni Cecchelli\textsuperscript{b,c}\par Edward Plumb\textsuperscript{b}\quad Maryam Kamgarpour\textsuperscript{a} \par Galit Ashkenazi-Golan\textsuperscript{b}}

{\small
\textsuperscript{a}\, SYCAMORE Lab, EPFL - Swiss Federal Technology Institute of Lausanne, Lausanne, Switzerland\\
\textsuperscript{b}\, Department of Mathematics, London School of Economics and Political Science, London, UK\\
\textsuperscript{c}\, Department of Quantitative Methods, CUNEF Universidad, Madrid, Spain\\
\par}

\end{center}

\bigskip

\begin{abstract}
While infinitely repeated games admit a rich set of Nash equilibria, finitely repeated games typically have a much smaller and often inefficient one. We show how to enlarge this set using deposits: in each period a player may place a refundable sum with a neutral intermediary, returned when the game ends and forfeited following a deviation. Paying these deposits is voluntary and incentive compatible at every stage, so no commitment by the players is assumed, the only commitment required being that of the intermediary to a refund rule fixed before play begins. The mechanism sustains payoff profiles more efficient than those of the standard equilibria, without altering the underlying game and without transfers between players. We demonstrate it on the prisoner's dilemma, a congestion game, and a public goods game, all settings where cooperation cannot emerge in the standard finitely repeated version. We also apply it to a dynamic common-pool resource, suggesting that the construction extends beyond repeated stage-games.
\end{abstract}

% Keywords: EJOR requires 1-5 keywords, the first taken from the official
% list of EJOR keywords (here "Game theory", handled by Editor E. Borgonovo).
% The same keywords, in the same order, must be entered in the submission
% system, with the first one also entered under "Section/Category".
\noindent\textbf{Keywords:} Game theory; Repeated games; Cooperation;
Mechanism design; Subgame perfect equilibrium

\bigskip
%\thispagestyle{empty}
%\pagestyle{empty}

%%%%%%%%%%%%%%%%%%%%%%%%%%%%%%%%%%%%%%%%%%%%%%%%%%%%%%%%%%
\section{Introduction}

% Introducing the infinite horizon setting and the folk theorem
One of the central insights of game theory is that repeated interactions can sustain cooperation. The ability to reply to the actions of one's opponents, to punish undesirable behaviour and reward a desired one, gives rise to a richness of equilibrium outcomes in an infinitely repeated interaction, and to the ability to support cooperative outcomes. This richness is captured by the folk theorem, which characterises the set of equilibrium payoffs in infinitely repeated games. Namely, any feasible payoff vector that guarantees each player at least their minmax value, that is, the best payoff a player can secure against any behaviour of the opponents by responding optimally to it, can be supported as a Nash equilibrium, provided players are sufficiently patient \citep{sorin1992repeated, fudenberg1986folk}. This result has profoundly shaped the analysis of cooperation in economics, from collusion and relational contracts to international agreements \citep{green1984noncooperative, levin2003relational, barrett1994self}.

% Highlight difference with the finitely repeated setting
This richness emerging in infinite interactions stands in contrast to the set of equilibrium payoffs in finitely repeated games. When the horizon is fixed, backward induction may restrict the set of equilibrium outcomes \citep{fudenberg1991game}. This limitation can be overcome when the stage-game that is repeated has multiple Nash equilibria \citep{benoit1984finitely, krishna1985finitely, gossner2020repeated}. However, in settings such as the repeated Prisoner’s Dilemma with a unique stage-game Nash equilibrium, cooperation fails. In the last period, players face what is effectively a one-shot game, so the only equilibrium play is the stage-game Nash equilibrium. Since this outcome is fixed regardless of prior history, the same reasoning applies to the second-to-last period, and, by backward induction, to every period before it. Thus, while the infinitely repeated version admits a continuum of Nash equilibria, the only Nash equilibrium of the finitely repeated version is the stage-game Nash equilibrium repeated at each iteration. The gap is not confined to stylised examples: in operational settings such as repeated inventory transshipment, full cooperation is a subgame perfect equilibrium of the infinitely repeated game only when players are sufficiently patient \citep{huang2010repeated}, and the argument does not survive a known terminal date. A subgame perfect equilibrium is one that remains an equilibrium in the subgame following every history, including the histories that arise only when some player has deviated.

A large body of work has attempted to bridge the gap between finitely and infinitely repeated games. Some of the most common approaches introduce incomplete information \citep{kreps1982rational} or uncertainty about the horizon \citep{normann2012impact}. Other studies assume limited rationality \citep{radner1986can}, bounded complexity of the strategies \citep{neyman1985bounded, rubinstein1986finite}, or add external enforcement that implements commitment \citep{faina1998unilateral, kalai2010commitment} or removes players, namely ostracism \citep{hirshleifer1989cooperation}. A related strand supports an outcome that is not itself an equilibrium by means of incentive strategies, and asks under which conditions the implied threat is credible \citep{breton2008incentive}. A further strand leaves the game untouched and weakens the solution concept instead, asking only that no player gain more than $\varepsilon$ by deviating \citep{radner1981monitoring, parilina2015approximated, flesch2016refinements}. Here $\varepsilon$ is the payoff a player must be willing to forgo for cooperation to persist, and it is not recovered. While these approaches can sustain cooperation, they do so by modifying the information available to the players, the horizon, or other parameters, by requiring some enforcing procedure, or by relaxing the equilibrium notion itself. As a result, it remains unclear whether the inefficiency of the original finitely repeated game reflects a fundamental limitation of finite interaction.

We introduce a deposit mechanism that applies to arbitrary games and requires no commitment from the players themselves. What we require is a reliable deposit mechanism, which is a neutral intermediary rather than any participant that commits to it. At each period every agent may place a sum of money with a neutral intermediary, and the accumulated stake is returned when the game ends. Paying is voluntary, and in equilibrium each player chooses to pay in their own interest. The mechanism does not alter the underlying stage-game and requires no transfers across agents. It recovers, and in fact extends, the payoffs of the Nash-threat folk theorem \cite{friedman1971non}, the version of the folk theorem in which a deviation is punished by permanent reversion to a Nash equilibrium of the stage-game rather than by a minmax punishment. The mechanism sustains, as a subgame perfect equilibrium of the finitely repeated game, any feasible payoff vector that gives every player strictly more than the lowest payoff that player receives across all Nash equilibria of the stage-game. This set contains every payoff that strictly Pareto-dominates a single stage-game Nash equilibrium (as in \citealp{fudenberg2007nash}), but is generally larger. In this way, the mechanism can sustain cooperation.

% Presentation of the main result
Our main result shows that the set of equilibrium outcomes in finitely repeated games can be significantly enlarged with the above proposed mechanism. In particular, any payoff profile that can be supported by an equilibrium of the infinitely repeated game which uses Nash threats can be implemented as a subgame perfect equilibrium of a finitely repeated game augmented with deposits. The construction applies to arbitrary stage-games and sufficiently long but finite horizons. 

%The mechanism assumes the possibility to have a device where players deposit sums, and can apply pre-determined rules as to when and how much of the sums are paid back to the players. The mechanism observes the deposits the players pay in every period, and is able to communicate it to all players. Besides the ability of deposit, the sets of actions available to the players remain the same, the game retains its information structure, and there is no additional assumption regarding the ability to form groups or ostracise players. 

% Discounted cost of the deposits vs accumulation growing with time
%The result is driven by a simple but powerful asymmetry. Deposits are costly because they are paid before the terminal date and, in our baseline formulation, do not earn interest while held by the intermediary. However, the discounted cost of deposits remains uniformly bounded, their undiscounted accumulation grows linearly with the horizon. This allows agents to transform a sequence of small, incentive-compatible payments into a large terminal stake that disciplines behaviour even in the final periods. The mechanism effectively recreates, in finite time, the continuation values that sustain cooperation in infinite-horizon models.  

%When agents do not discount the future, using deposits entails no cost, as any payment is fully recovered at the terminal date. In this case, sustaining cooperation is straightforward: agents can accumulate arbitrarily large stakes over time without affecting their payoffs, making deviations unprofitable even in the final periods.

In a repeated game setting, the players typically evaluate a stream of payoffs using a discount factor: obtaining a sum today is better than obtaining the same sum tomorrow. We do not assume that deposits earn interest or otherwise gain value over time, and we show that the mechanism works even when they do not. The reason is a contrast between two quantities. The present-value cost of the deposits stays bounded no matter how long the game lasts, while the total amount accumulated grows in proportion to the horizon. A player can therefore build up a large terminal stake out of a sequence of small payments, at a present-value cost that stays small.

Deposit–refund and hostage mechanisms have been proposed as a way to sustain cooperation, particularly in the context of international environmental agreements \citep{hovi2012credible, hovi2015hope, werner1993hostages, mcevoy2025providing, fuentesalbero2010can}. They differ from ours in two main respects. First, they focus on a specific application and to a single cooperative target, establishing that this one outcome can be supported, rather than characterising the whole set of profiles a mechanism can sustain for a class of games. Second, they place a single stake before play begins and do not involve discounting. Once future payoffs are discounted, however, such an upfront stake is the most expensive way to post collateral, since a sum committed today is worth more than the same sum returned later. In our setting the stake can be spread over the horizon rather than posted upfront, so its timing becomes a design choice.

The closest antecedent is the payment-scheme approach of \citet{parilina2024payment}, but the mechanisms are still very different. There, an intermediary receives the players' rewards and withholds the cooperative surplus at the source; in our setting the players receive the reward and then decide whether to make a deposit, a choice that is incentive compatible at every stage and requires no transfers between players. Their analysis also fixes a single target, the joint-payoff-maximising profile, and asks how to sustain it, whereas we characterise the entire set of supportable outcomes. Finally, where their schemes are cast as individually rational, in the sense that no player does worse than by reverting to non-cooperative play, and stable against deviation for a cooperative imputation, ours delivers a subgame perfect equilibrium of the finitely repeated game.

Our result offers another perspective on cooperation in finitely repeated games. Alongside approaches that relax the horizon or the informational structure, cooperation can also be sustained by introducing instruments that let agents transfer incentives across time. Such instruments are of interest wherever interaction is finite and long-term commitment is unavailable, including repeated interactions among learning algorithms, where cooperative outcomes have been argued to reflect incomplete exploration rather than a stable equilibrium \citep{abada2024collusion}.

The rest of the paper is organised as follows. Section \ref{sec:model} introduces the stage-game, the augmented finitely repeated game with deposits, and the infinite-horizon benchmark whose outcomes we aim to recover. Section \ref{sec:main} states and proves our main result. Section \ref{sec:optimal} shows how to compute the cheapest deposit scheme for a given target and horizon, and how the same program accommodates further requirements. Section \ref{sec:examples} demonstrates the mechanism on four examples, and Section \ref{sec:conclusion} concludes. The first example is the repeated prisoner's dilemma, for which we show that horizons as short as two or three periods can be enough to sustain cooperation with the mechanism proposed. The second, a congestion game with a rotating allocation, walks concretely through how the mechanism operates and how small its implementation cost can be. The third, a public goods game, shows that the mechanism does not necessarily need to rely on the harshest available punishment. The fourth, a dynamic common-pool resource with an evolving stock, extends the mechanism beyond repeated stage-games.
% Possible applications
%Applications of the mechanism arise naturally in environments where interactions are finite, enforcement is limited, and some form of neutral intermediary is already available. In international settings such as climate and/or trade agreements, such as the Paris Agreement \citep{agreement2015paris}, countries could pledge financial assets or sovereign bonds that are returned conditional on compliance, strengthening incentives without requiring direct transfers. More broadly, the mechanism is directly implementable in digital environments where smart contracts can hold and conditionally release funds \citep{kolvart2016smart}. This includes online markets \citep{ba2003building} and emerging AI-driven economies \citep{hadfield2025economy, tomasev2025virtual}, where agents interact repeatedly but without long-term commitment, and where smart contracts can provide a minimal enforcement technology.

%%%%%%%%%%%%%%%%%%%%%%%%%%%%%%%%%%%%%%%%%%%%%%%%%%%%%%%%%%
\section{Environment and mechanism}\label{sec:model}

\subsection{The stage-game}

We consider a stage-game $G = (N, \{A^i\}_{i \in N}, \{r^i\}_{i \in N})$, where $N$ is a finite set of players, with $n := |N|$, $A^i$ is the finite set of actions available to player $i$, and $r^i : A \to \mathbb{R}$ is the stage payoff function associated with player $i$, with $A:=\prod_{i \in N} A^i$ the set of action profiles. We write $\Delta(A^i)$ for the set of mixed actions of player $i$ and $\Delta := \prod_{i \in N} \Delta(A^i)$ for the set of mixed actions. 
As mixed actions randomise over pure actions, and by linearity of expectation, we can, via a slight abuse of notation, extend the payoff function multilinearly to any $\ra \in \Delta$ by
\begin{align*}
    r^i(\ra) := \mathbb{E}_{a\sim \ra}[r^i(a)]= \sum_{a \in A} \Big( \prod_{j \in N} {\ra}^j(a^j) \Big)\, r^i(a).
\end{align*}
An action profile $\ra \in \Delta$ is a Nash equilibrium of $G$ if $r^i(\ra) \ge r^i({\rta}^i, {\ra}^{-i})$ for every player $i$ and every ${\rta}^i \in \Delta(A^i)$, where ${\ra}^{-i} := ({\ra}^j)_{j \neq i}$. Let $\rN \subseteq \Delta$ denote the set of stage-game Nash equilibria. By Nash's theorem $\rN$ is non-empty, and as a closed subset of the compact set $\Delta$ it is compact, so the worst equilibrium payoff for player $i$ is well defined:
\begin{align*}
    r^{i,\mathrm{NE}} := \min_{\ra \in \rN} r^i(\ra).
\end{align*}
Denote by ${\ra}^{\mathrm{pun}}(i)$ a stage-game Nash equilibrium that attains this payoff for player $i$ (if there are several such equilibria, any of them will do). Formally,
${\ra}^{\mathrm{pun}}(i) \in \operatorname*{arg\,min}_{\ra \in \rN} r^i(\ra)$. We allow ${\ra}^{\mathrm{pun}}(i)$ to be mixed, since the equilibrium minimising a given player's payoff need not be in pure actions.

\subsection{The augmented finitely repeated game}

We enrich the standard definition of finitely repeated games with the concept of \emph{deposits} and \emph{terminal payoffs}. We first fix non-negative deposits $(d^i_t)_{i\in N, t\in [T]}\in \mathbb{R}^{n\times T}_{\geq 0}$, where $[T] := \{1,\dots,T\}$. At period $t$, after seeing the realised action profile and receiving their payoff, player $i$ makes a binary decision to either pay a deposit $d^i_t$ or not. Upon termination of period $T$, players receive a terminal payoff given by the terminal payoff functions $(w^i)_{i\in N}$ defined below. More specifically, the game unfolds as follows at period $t \in [T]$:

\begin{enumerate}
    \item Each player $i\in N$ simultaneously chooses an action $a_t^i \in A^i$. Let $a_t = (a_t^i)_{i \in N}$ denote the action profile played at time $t$.
    \item Each player observes the full action profile and receives stage payoff $r^i(a_t)$.
    \item Each player $i$ then decides whether to pay deposit $d_t^i$. If player $i$ pays, their period-$t$ payoff is reduced by $d_t^i$. Each player observes the full deposit decisions of all players. Let $p_t^i=1$ if player $i$ paid the deposit at time $t$ and $p_t^i=0$ otherwise, and let $p_t = (p_t^i)_{i \in N}$.
    \item At the end of period $T$, each player $i$ receives a terminal payoff $w^i$, which may depend on the full history of play as defined below.
\end{enumerate}

Let $H_t$ denote the set $(A \times \{0,1\}^N)^t$ of histories of length $t$, recording all actions and deposit decisions up to and including period $t$, with $H_0 := \varnothing$ the empty history. Each terminal payoff function is then of the form $w^i:H_T\to \mathbb{R}_{\geq 0}$. 

Players discount future payoffs by a common factor $\delta \in (0,1)$. The total discounted payoff to player $i$ along a terminal history $h_T = (a_t, p_t)_{t=1}^T$ is
\begin{align*}
    U^i(h_T) = \sum_{t=1}^T \delta^{t-1} \left( r^i(a_t) - p_t^i d_t^i \right) + \delta^{T-1} w^i(h_T).
\end{align*}

We write $G_\delta^T(d,w)$ for the augmented finitely repeated game with horizon $T$, fixed deposit scheme $d = (d^i_t)_{i \in N,\, t \le T}$, terminal-payoff rule $w = (w^i)_{i \in N}$ and discount factor $\delta$. The game $G_\delta^T(d,w)$ has $N$ as the set of players. A behaviour strategy of a player in the finitely repeated game determines both the action to take and the deposit decisions. Formally, a behaviour strategy for player $i$ is a pair $(\sigma^i, \rho^i)=((\sigma_t^i,\rho^i_t)_{t=1,\ldots,T})$ defined as follows:
\begin{itemize}
\item player $i$ samples $a^i_t$ according to a distribution $\sigma_t^i : H_{t-1} \to \Delta(A^i)$, the \emph{action rule}. We allow players to sample their actions since stage-game Nash equilibria used as the threat in the Nash-threat folk theorem may require randomisation. 
\item player $i$ samples $p^i_t$ according to $\rho_t^i : H_{t-1} \times A \to \{0,1\}$, the \emph{deposit rule}, which prescribes a deposit decision, given the past history and the last action profile $a_t$ played. For simplicity, we consider only pure deposit decisions\footnote{This is without loss of generality, as not depositing will be punished immediately.}.
\end{itemize}
A strategy profile $(\sigma, \rho):=(\sigma^i, \rho^i)_{i\in N}$ induces a distribution $\mathbb{P}_{(\sigma,\rho)}$ over terminal histories $H_T$, and we write $\mathbb{E}_{(\sigma,\rho)}$ for the corresponding expectation. Player $i$ evaluates a profile $(\sigma,\rho)$ by $\mathbb{E}_{(\sigma,\rho)}\big[ U^i(h_T) \big]$. For any history $h\in \bigcup_{t=0}^{T-1} H_t$, a strategy-profile $(\sigma,\rho)$ induces a distribution over the terminal histories $H_T$, and the respective expected reward for player~$i$ is $\mathbb{E}_{(\sigma,\rho)|h}\big[ U^i(h_T) \big]$. A subgame perfect equilibrium in $G^T_\delta(d,w)$ is a strategy profile $(\sigma^*,\rho^*)$ such that for every $i\in N$, strategy for player $i$, $(\sigma^i, \rho^i)$ and history $h\in \bigcup_{t=0}^{T-1} H_t$, it holds that $\mathbb{E}_{(\sigma^*,\rho^*)|h}\big[ U^i(h_T) \big]\geq \mathbb{E}_{(({\sigma^*}^{-i},{\rho^*}^{-i}), (\sigma^i,\rho^i))|h}\big[ U^i(h_T) \big]$, where $(({\sigma^*}^{-i},{\rho^*}^{-i}), (\sigma^i,\rho^i))$ is the strategy profile where each player $j\in N\setminus\{i\}$ plays $({\sigma^*}^j, {\rho^*}^j)$, and player $i$ plays $(\sigma^i,\rho^i)$.  

It is natural to require that the terminal payoffs can be funded by the deposits paid during the finitely repeated game. Hence, we impose the following feasibility condition: for every terminal history $h_T \in H_T$ and every player $i$,
\begin{align*}
        w^i(h_T) \;\leq\; \sum_{t=1}^{T} p_t^i\, d_t^i.
\end{align*}
Thus, we require no cross-subsidisation between players and no external funding. In practice, one might have an intermediary whose sole function is to hold deposits and return them according to the predetermined rule $w^i$ (this intermediary is not a strategic player in this game, but rather a facilitating mechanism). The concept of \emph{terminal payoffs} was introduced by \citet{kandori1992repeated}, with the difference that here the terminal payoff is funded by the players' own deposits rather than being exogenous.

\subsection{The infinite-horizon benchmark}

In a finitely repeated game, backward induction impedes cooperation: a deviation in the final round cannot be punished, so defection is optimal there, and the argument propagates backwards to the first round. However, in infinitely repeated games, every round has a future. A deviation can therefore be punished by the other players, who switch to actions that lower the deviator's continuation payoff. When players are sufficiently patient, this threat sustains cooperation as an equilibrium. Our setting recreates the same threat in finite time. A deposit that is returned only if no one deviates acts like a continuation payoff, so losing it after a deviation lowers the deviator's payoff just as a reduced continuation value would, recreating the incentive to cooperate. We consider infinite sequences of pure action-profiles. We provide conditions for the first $T$ action-profiles of such sequences to be the equilibrium play of the $T$-period game with deposits. In Sections~\ref{subsec:folk} and~\ref{subsec:folk-finite} we show, using previous results from game theory, that this set of sequences is rich, and obtain a folk-theorem result.

Let $G^\infty_\delta$ denote the infinitely repeated game with stage-game $G$ and discount factor $\delta \in (0,1)$. Write
\begin{align*}
    A^{\mathbb{N}} := \big\{ \mathbf{a} = (a_1, a_2, \ldots) \;:\; a_t \in A \ \text{ for every } t \ge 1 \big\}
\end{align*}
for the set of infinite sequences of pure action-profiles, and consider such a sequence $\mathbf{a} \in A^{\mathbb{N}}$. The utility for player~$i$ from this sequence is

\begin{align*}
    (1 - \delta)\, \mathbb{E}\Big[ \sum_{t=1}^\infty \delta^{t-1} r^i(a_t) \Big].
\end{align*}

We find sufficient conditions under which we can construct an equilibrium of the $T$-period game, where $a_{1:T}:=(a_1,\ldots,a_T)$ is the action history played in that equilibrium. To this end, we consider possible gains from unilateral deviations from this sequence.

For a given $\mathbf{a}=(a_1,a_2,\ldots)$, let $a_{1:t}:=(a_1, \ldots, a_t)$ denote its truncation at period $t$. Thus $a_{1:t}\in A^t$ is an action history, in contrast to $h_t\in H_t$, which records both actions and deposit decisions. Define $r^{i,\mathrm{coop}}_t$ to be the stage-game payoff for player $i$ if the action profile $a_t$ is played, and  $r_t^{i,\mathrm{dfct}}$  to be the maximal stage-game payoff for a unilateral deviation of player~$i$ from action profile $a_t$.
\begin{alignat*}{3}
    &r_t^{i,\mathrm{coop}} &&= r^i(a_t), \\
    &r_t^{i,\mathrm{dfct}} &&= \max_{a^i \in A^i\setminus \{a^i_t\}} r^i\big( a^i,a^{-i}_t\big).
\end{alignat*}

In the finite-horizon game equilibria we introduce, deviations are punished by a combination of using the deposits, and by switching to playing the stage-game Nash equilibrium with lowest payoff for the deviating player forever. Therefore, we focus on the class of action-profile sequences that can be supported by Nash-threats with a strict incentive margin. These are sequences such that the continuation payoff for each player, when following the sequence, exceeds the Nash-threat continuation payoff by a uniform positive amount. 

Formally, define $\mathcal{A}^\ast \subseteq A^{\mathbb{N}}$ as the set of sequences $\mathbf{a} \in A^{\mathbb{N}}$ such that for every player $i \in N$ there exists $\varepsilon^i > 0$ satisfying, for all $t \in \{1, 2, \ldots\}$,
\begin{align}\label{eq:sigma-star}
    \sum_{\tau=t}^{\infty} \delta^{\tau-t} r^{i,\mathrm{coop}}_\tau
\;\ge\;
\max\big\{ r_t^{i,\mathrm{coop}},\, r_t^{i,\mathrm{dfct}} \big\} + \sum_{\tau=t+1}^{\infty} \delta^{\tau-t} r^{i,\mathrm{NE}} + \varepsilon^i.
\end{align}

Equivalently,
\begin{align*}
    \mathcal{A}^\ast = \Big\{\, \mathbf{a} \in A^{\mathbb{N}} \;:\; \forall i \in N \ \ \exists\, \varepsilon^i > 0 \ \ \forall t \ge 1, \ \eqref{eq:sigma-star} \text{ holds} \,\Big\} .
\end{align*}

The left-hand side is the continuation value of the sequence from period $t$ onwards. The right-hand side is the value of the most profitable stage-game deviation in period $t$, followed by permanent reversion to the worst stage-game Nash equilibrium for the deviator. When $r_t^{i,\mathrm{dfct}} \ge r_t^{i,\mathrm{coop}}$ the maximum is $r_t^{i,\mathrm{dfct}}$. When instead the prescribed action is itself a stage best reply, the binding term is $r_t^{i,\mathrm{coop}}$ and the condition reduces to $\sum_{\tau=t+1}^{\infty} \delta^{\tau-t}\big(r^{i,\mathrm{coop}}_\tau - r^{i,\mathrm{NE}}\big) \ge \varepsilon^i$, a participation requirement ensuring that the cooperative continuation strictly dominates the Nash threat for every player. This inequality enables the construction of the finite-horizon equilibrium, where the action history played (in equilibrium) is the truncation $a_{1:T}$. The equilibrium construction uses the one-shot deviation principle, according to which a strategy profile is a subgame perfect equilibrium of a finite-horizon game precisely when no player can gain by departing from it at a single history and conforming thereafter.

\subsection{Which payoffs does the class $\mathcal{A}^\ast$ generate?}\label{subsec:folk}

The class $\mathcal{A}^\ast$ is defined by an incentive condition, and it is not immediate from that definition which payoff vectors its members generate, nor indeed that it is non-empty. This subsection settles the question using results from the literature on infinitely repeated games: once players are patient enough, $\mathcal{A}^\ast$ generates every feasible payoff vector that gives each player strictly more than their worst stage-game Nash payoff.

Write
\begin{align*}
    U := \{ r(a) : a \in A \} \subset \mathbb{R}^{N}, \qquad
    F := \operatorname{co} U, \qquad
    F^{\mathrm{NE}} := \big\{ v \in F : v^i > r^{i,\mathrm{NE}} \ \text{for every } i \in N \big\},
\end{align*}
so that $U$ collects the payoff vectors generated by pure action profiles, $F$ is the feasible set, and $F^{\mathrm{NE}}$ is the set of feasible payoffs that are individually rational relative to the threats the mechanism uses. We also write
\begin{align*}
    C^i := \max_{a \in A} r^i(a) - \min_{a \in A} r^i(a)
\end{align*}
for the range of player $i$'s stage payoff, a constant that depends neither on the discount factor nor on the target sequence. The reference vector $(r^{i,\mathrm{NE}})_{i \in N}$ need not be the payoff vector of any single stage-game Nash equilibrium, since distinct players may be minimised by distinct equilibria. Consequently $F^{\mathrm{NE}}$ contains, and in general strictly contains, the set of feasible payoffs that Pareto dominate the payoffs of one fixed stage-game Nash equilibrium.

The bridge from a payoff vector to a sequence of pure action profiles is supplied by the literature. \citet{sorin1986repeated} showed that, for patient enough players, every point of $F$ is exactly the discounted average of a deterministic sequence of pure action profiles; this statement appears as Lemma~1 of \citet{fudenberg1991dispensability}, who credit it to Sorin and record the sufficient bound $\delta \ge 1 - 1/L$ when $F$ is a polytope with $L$ extreme points. Exact attainability alone does not suffice for us, because \eqref{eq:sigma-star} constrains every continuation of the sequence and not merely its value at the first date, and the continuation payoffs along an arbitrary attaining sequence may wander far from the target. What we use is the strengthening given by Lemma~2 of \citet{fudenberg1991dispensability}, restated here in our notation.

\begin{lemma}[Lemma~2 of \citet{fudenberg1991dispensability}]\label{lem:FM}
For every $\eta > 0$ there exists $\underline{\delta}_\eta < 1$ such that for every $\delta \in (\underline{\delta}_\eta, 1)$ and every $v \in F$ there is a sequence $\mathbf{a} = (a_1, a_2, \ldots)$ of pure action profiles satisfying
\begin{align*}
    (1-\delta) \sum_{t = 1}^{\infty} \delta^{t-1} r(a_t) = v
    \qquad \text{and} \qquad
    \Big\lVert (1-\delta) \sum_{\tau = t}^{\infty} \delta^{\tau - t} r(a_\tau) - v \Big\rVert < \eta
    \ \ \text{for every } t \ge 1.
\end{align*}
\end{lemma}

Public randomisation is therefore dispensable: the target payoff is delivered by a deterministic path of pure action profiles whose continuation values never drift far from it, which is exactly the input \eqref{eq:sigma-star} requires. Proposition~\ref{prop:folk} is an immediate consequence of Lemma~\ref{lem:FM}.

\begin{proposition}\label{prop:folk}
Let $v \in F^{\mathrm{NE}}$ and set $\eta := \tfrac{1}{2} \min_{i \in N} \big( v^i - r^{i,\mathrm{NE}} \big) > 0$. There exists $\underline{\delta} < 1$ such that for every $\delta \in (\underline{\delta}, 1)$ there is a sequence $\mathbf{a} \in \mathcal{A}^\ast$ whose normalised discounted payoff vector is exactly $v$ and all of whose continuation payoff vectors lie within $\eta$ of $v$.
\end{proposition}

\begin{proof}
Let $\underline{\delta}_\eta$ be as in Lemma~\ref{lem:FM}, fix $\delta > \underline{\delta}_\eta$, and let $\mathbf{a}$ be the sequence Lemma~\ref{lem:FM} provides for $v$ and $\eta$, so that $r^{i,\mathrm{coop}}_t = r^i(a_t)$ for every $i$ and $t$. The final claim of the proposition is then the continuation bound of Lemma~\ref{lem:FM}, and it remains to check that $\mathbf{a} \in \mathcal{A}^\ast$. Applying that bound at date $t+1$,
\begin{align*}   (1-\delta) \sum_{\tau = t+1}^{\infty} \delta^{\tau - t - 1} r^{i,\mathrm{coop}}_\tau   \;\ge\; v^i - \eta \;\ge\; r^{i,\mathrm{NE}} + \eta ,\end{align*}
the second inequality because $v^i - r^{i,\mathrm{NE}} \ge 2\eta$. Equivalently, for every $i \in N$ and every $t \ge 1$,
\begin{align}\label{eq:gap-lower}   \sum_{\tau = t+1}^{\infty} \delta^{\tau - t - 1} \big( r^{i,\mathrm{coop}}_\tau - r^{i,\mathrm{NE}} \big)   \;\ge\; \frac{\eta}{1-\delta} .\end{align}
Now rewrite \eqref{eq:sigma-star}. Its left-hand side equals $r^{i,\mathrm{coop}}_t + \delta \sum_{\tau \ge t+1} \delta^{\tau-t-1} r^{i,\mathrm{coop}}_\tau$, its Nash-reversion term equals $\delta\, r^{i,\mathrm{NE}}/(1-\delta)$, and
$\max\{ r^{i,\mathrm{coop}}_t, r^{i,\mathrm{dfct}}_t \} = r^{i,\mathrm{coop}}_t + \max\{ 0,\, r^{i,\mathrm{dfct}}_t - r^{i,\mathrm{coop}}_t \}$.
Cancelling $r^{i,\mathrm{coop}}_t$ from both sides shows that \eqref{eq:sigma-star} holds at date $t$ if and only if
\begin{align}\label{eq:sigma-star-equiv}    \delta \sum_{\tau = t+1}^{\infty} \delta^{\tau - t - 1} \big( r^{i,\mathrm{coop}}_\tau - r^{i,\mathrm{NE}} \big)    \;\ge\; \max\big\{ 0,\, r^{i,\mathrm{dfct}}_t - r^{i,\mathrm{coop}}_t \big\} + \varepsilon^i \end{align}
By \eqref{eq:gap-lower} the left-hand side of \eqref{eq:sigma-star-equiv} is at least $\delta \eta / (1-\delta)$, uniformly in $t$, while the first term on the right-hand side is at most $C^i$, uniformly in $t$ and independently of $\delta$. Enlarge $\underline{\delta} \ge \underline{\delta}_\eta$ so that $\delta \eta/(1-\delta) > C^i$ for every $i \in N$, which is possible since $\delta/(1-\delta) \to \infty$ as $\delta \to 1$, and set $\varepsilon^i := \delta \eta/(1-\delta) - C^i > 0$. Then \eqref{eq:sigma-star-equiv}, and hence \eqref{eq:sigma-star}, holds at every date with the margin $\varepsilon^i$, which does not depend on $t$. Therefore $\mathbf{a} \in \mathcal{A}^\ast$, and its normalised discounted payoff is $v$ by the first part of Lemma~\ref{lem:FM}.
\end{proof}

Proposition~\ref{prop:folk} is a pure-strategy sharpening, with deviator-specific threats, of a classical result. Theorem~C of \citet{fudenberg1986folk}, attributed there to \citet{friedman1971non}, states that if a feasible, individually rational payoff vector $v$ Pareto dominates the payoff vector of a one-shot Nash equilibrium $(e_1, \ldots, e_n)$ of the stage-game, then for $\delta$ close enough to one the infinitely repeated game has a perfect equilibrium with average payoff $v$, sustained by exactly the strategies we have in mind: play the actions that sustain $v$ until someone deviates, and play $(e_1, \ldots, e_n)$ forever thereafter. Indeed, combining Proposition~\ref{prop:folk} with the one-shot deviation principle recovers that conclusion, since every $\mathbf{a} \in \mathcal{A}^\ast$ makes the grim Nash-threat profile a subgame perfect equilibrium of $G^\infty_\delta$.

\section{Main result}\label{sec:main}

Theorem~\ref{thm:main} to follow is the central result of the paper. It establishes that, for any infinite-horizon action sequence $\mathbf{a}^\ast \in \mathcal{A}^\ast$ that satisfies condition~\eqref{eq:sigma-star}, there is a subgame perfect equilibrium of the $T$-period augmented game with deposits and terminal payoffs, whose equilibrium play is the truncated action history $a^\ast_{1:T}$. Specifically, for all $\mathbf{a}^\ast\in \mathcal{A}^\ast$ and all sufficiently long horizons $T$, we exhibit a deposit scheme $d$ and a terminal-payoff rule $w$ under which playing $a^\ast_{1:T}$ and paying the prescribed deposit at each period is the outcome of a subgame perfect equilibrium of $G^T_\delta(d, w)$.

To define the candidate equilibrium we first need to identify, at any full history $h_t\in H_t$, whether the play has left the equilibrium path and, if so, who left it first. Fix a target action sequence $\mathbf{a}^\ast \in \mathcal{A}^\ast$, and let $h_t^\ast$ denote the on-path full history, i.e. the unique history of length $t$ generated when every player plays the target actions $a^\ast_{1:t}$ and pays the prescribed deposits. Denote by $a^{\ast,i}_t$ the action prescribed to player $i$ at period $t$ according to $\mathbf{a}^\ast$. For an arbitrary full history $h_t\in H_t$, let $i(h_t)$ denote the first player to deviate from the target action sequence or from a prescribed deposit, with ties broken by lowest index. If no player has deviated, then $h_t = h_t^\ast$ and we say $h_t$ is on path.

For a finite time horizon $T$, the equilibrium action strategy profile $\tilde{\sigma}=\tilde{\sigma}_T$ plays the actions of $a^\ast_{1:T}$ on the equilibrium path and switches permanently to ${\ra}^{\mathrm{pun}}(i(h_t))$, the stage-game Nash equilibrium minimising the first deviator's payoff, after any deviation. For every history $h_t$ with $t \leq T$,
\begin{align*}
\tilde{\sigma}_t(h_t) =
\begin{cases}
a^{\ast,i}_t & \text{if } h_t = h_t^\ast, \\[2pt]
{\ra}^{\mathrm{pun}}(i(h_t)) & \text{otherwise.}
\end{cases}
\end{align*}
The deposit rule $\tilde{\rho}$ prescribes the payment of the deposit on the equilibrium path and no payment once any deviation has occurred:
\begin{align*}
\tilde{\rho}_t(h_t) =
\begin{cases}
1 & \text{if } h_t = h_t^\ast, \\[2pt]
0 & \text{otherwise.}
\end{cases}
\end{align*}
The terminal-payoff function returns the reward $\bar{w}^i$ to every player except the first deviator, capped at the deposits that player has themselves paid. Writing $i(h_T)$ for the first deviator, and for some $\bar{w}^i \ge 0$ described in the following proof, we set
% CHANGED (M7): align* -> align, label restored; rejected alternative was an unnumbered align*
\begin{align} \label{eq:terminal_payoff}
w^i(h_T) =
\begin{cases}
\bar{w}^i & \text{if } h_T = h_T^\ast, \\[2pt]
0 & \text{if } h_T \neq h_T^\ast \text{ and } i = i(h_T), \\[2pt]
\min\Big\{ \bar{w}^i,\; \sum_{t=1}^{T} p^i_t d^i_t \Big\} & \text{if } h_T \neq h_T^\ast \text{ and } i \neq i(h_T).
\end{cases}
\end{align}
On the equilibrium path each player pays $d_t^i$ in every period, so the accumulated stake is $\sum_t d_t^i$ and the feasibility condition $w^i(h_T) \le \sum_{t=1}^T p_t^i d_t^i$ reduces to $\bar{w}^i \le \sum_t d_t^i$. On any other history the refund is capped at the deposits the player has actually paid, so feasibility is immediate there too. In particular the mechanism is self-funded on every history, no player forfeits their stake on account of another player's deviation, and no player receives more off the equilibrium path than on it.

\begin{theorem} \label{thm:main}
For $\delta \in (0,1)$, let $G^\infty_\delta$ be an infinitely repeated game. Let $\mathbf{a}^\ast\in\mathcal{A}^\ast$ be a sequence of action-profiles of $G^\infty_\delta$, and let $(\varepsilon^i)_{i \in N}$ be positive numbers satisfying condition \eqref{eq:sigma-star}. There is a finite threshold $T^\ast$, depending only on the stage-game $G$, the discount factor $\delta$ and the margins $(\varepsilon^i)_{i \in N}$, such that for every horizon $T \ge T^\ast$ there exist a deposit scheme $d=(d^i_t)_{i\in N, t\in [T]}$ and a terminal-payoff rule $w = (w^i)_{i \in N}$ for which the truncated profile $(\tilde{\sigma}, \tilde{\rho})$ is a subgame perfect equilibrium of $G_{\delta}^T(d, w)$. 

%We also note that, along the equilibrium path, $(\tilde{\sigma}, \tilde{\rho})$ induces the same action profile as $\sigma$ at every period $t \le T$.
\end{theorem}

\begin{proof}
Throughout, fix $\mathbf{a}^\ast\in\mathcal{A}^\ast$ and margins $(\varepsilon^i)_{i \in N}$ satisfying \eqref{eq:sigma-star}. We first define several constants, none of which depends on the horizon $T$, and then verify subgame perfection through the one-shot deviation principle.

For each player $i \in N$ we set
\begin{align*}
\bar{r}^i := \max_{a \in A} |r^i(a)|, \qquad C^i := \max_{a \in A} r^i(a) - \min_{a \in A} r^i(a),
\end{align*}
where $C^i$ is as defined in Section~\ref{subsec:folk}, so that $\bar{r}^i$ bounds every stage payoff and $|r^{i,\mathrm{coop}}_t - r^{i,\mathrm{NE}}| \le C^i$ for every $t$. Note that $C^i > 0$: if player $i$'s stage payoff were constant, both sides of \eqref{eq:sigma-star} would coincide and the condition would read $0 \ge \varepsilon^i$, contradicting $\varepsilon^i > 0$. Take the per-period deposit to be constant, $d^i_t := d^i := \frac{1}{2}(1-\delta)\,\varepsilon^i > 0$ for all $t$. We obtain that, for any $T$, the discounted cost of depositing is bounded as follows:
\begin{align}\label{eq:deposit-bound}
\sum_{\tau=t}^T \delta^{\tau-t} d^i \;\le\; \frac{d^i}{1-\delta} = \frac{\varepsilon^i}{2}, \qquad t \le T.
\end{align}
Note that the bound $\frac{\varepsilon^i}{2}$ is independent of $T$. Set
\begin{align*}
M^i := \max\left\{ 0,\; \left\lceil \frac{\log\!\big( 2 C^i / ((1-\delta)\varepsilon^i) \big)}{\log(1/\delta)} \right\rceil \right\},
\end{align*}
which is finite, independent of $T$, and satisfies $\delta^{M^i} \le (1-\delta)\varepsilon^i / (2C^i)$. Set $\bar{w}^i := 0$ if $M^i = 0$, and
\begin{align*}
\bar{w}^i := \delta^{-(M^i-1)}\left( \frac{2\bar{r}^i}{1-\delta} + \frac{\varepsilon^i}{2} \right) \quad \text{if } M^i \ge 1,
\end{align*}
and finally
\begin{align*}
T^\ast := \max_{i \in N} \max\left\{ M^i,\; \big\lceil \bar{w}^i / d^i \big\rceil \right\}.
\end{align*}
We now fix an arbitrary horizon $T \ge T^\ast$, and let $w$ be the terminal-payoff rule defined in Equation \eqref{eq:terminal_payoff}. We show that $(\tilde{\sigma}, \tilde{\rho})$ is a subgame perfect equilibrium of $G^T_\delta(d, w)$.

The game $G^T_\delta(d, w)$ is a finite multistage game with observed actions, so by the one-shot deviation principle \cite[Theorem 4.1]{fudenberg1991game}, $(\tilde{\sigma}, \tilde{\rho})$ is subgame perfect if and only if no player can profitably unilaterally deviate at one time-period while conforming thereafter. We rule out such deviations off the path and on the path in turn.

Consider a history $h_t$ that already contains a deviation, so $\tilde{\sigma}$ prescribes the stage-game Nash profile ${\ra}^{\mathrm{pun}}(i(h_t))$, while $\tilde{\rho}$ prescribes no payment, and each player's terminal refund depends only on their own deposit decisions and on whether they were the first to deviate. At the action node the prescribed action is player $i$'s component of a stage-game Nash equilibrium, hence a best reply to the other players' prescribed actions; since the continuation play and the terminal refund do not depend on the current action, no action deviation is profitable. At the deposit node, paying $d^i$ now raises the terminal refund by at most $d^i$, and any such increase arrives only at date $T$, so the payment costs at least $(1-\delta^{T-t})d^i \ge 0$ and withholding, as prescribed, is optimal. The cost is strictly positive at every $t < T$, and also at $t = T$ for the first deviator, whose refund is zero in any case. The single exception is a player other than the first deviator, at the final date, whose accumulated deposits, this payment included, would still not exceed $\bar{w}^i$, and who is then exactly indifferent between paying and withholding; since the one-shot deviation principle requires only weak optimality, the conclusion is unaffected. No off-path deviation is therefore profitable.

Now consider, for some $t \leq T$, a history $h^\ast_t$ that does not contain a deviation. An action deviation by player $i$ earns that player the stage payoff of some $a^i \neq a^{\ast,i}_t$ against $a^{\ast,-i}_t$, which is at most $r^{i,\mathrm{dfct}}_t$; the deviation is observed, so play reverts to ${\ra}^{\mathrm{pun}}(i)$ from period $t+1$ with zero terminal payoff, and this continuation is the same whether or not player $i$ pays the period-$t$ deposit. It is therefore strictly better for a player to deviate at both the action stage and the deposit stage than at the action stage alone, since withholding the deposit saves $d^i > 0$ and leaves the continuation unchanged. When considering action deviations, it thus suffices to deter this combined deviation, which earns to the deviator at most $r^{i,\mathrm{dfct}}_t$ in period $t$, forgoes the period-$t$ deposit, and yields continuation value $\sum_{\tau=t+1}^T \delta^{\tau-t} r^{i,\mathrm{NE}}$. If a player does not deviate at the action stage, but only at the deposit stage, they earn $r^{i,\mathrm{coop}}_t$, withhold the deposit, and trigger the same reversion to ${\ra}^{\mathrm{pun}}(i)$ with zero terminal payoff.

These two deviations share the continuation value $\sum_{\tau=t+1}^T \delta^{\tau-t} r^{i,\mathrm{NE}}$ and both forgo the period-$t$ deposit; they differ only in the period-$t$ stage payoff, which is $r^{i,\mathrm{dfct}}_t$ in the first and $r^{i,\mathrm{coop}}_t$ in the second. Writing $D^i_t := \max\{ r^{i,\mathrm{coop}}_t, r^{i,\mathrm{dfct}}_t \}$, the most profitable deviation in period $t$ is worth $D^i_t + \sum_{\tau=t+1}^T \delta^{\tau-t} r^{i,\mathrm{NE}}$, while conforming at both nodes yields the on-path value $\sum_{\tau=t}^T \delta^{\tau-t}\big( r^{i,\mathrm{coop}}_\tau - d^i \big) + \delta^{T-t}\bar{w}^i$. Deterring both deviations is therefore equivalent to the single inequality
\begin{align}\label{eq:IC}
\underbrace{\sum_{\tau=t}^T \delta^{\tau-t}\big( r^{i,\mathrm{coop}}_\tau - d^i \big)}_{\text{on-path payoff net of deposits}} \;+\; \underbrace{\delta^{T-t}\, \bar{w}^i}_{\text{terminal reward}} \;\ge\; \underbrace{D^i_t + \sum_{\tau=t+1}^T \delta^{\tau-t}\, r^{i,\mathrm{NE}}}_{\text{best one-shot deviation}}.
\end{align}

We establish \eqref{eq:IC} at every on-path period $t \le T$. We start by considering the truncated cooperation and deviation values, defined as follows       
\begin{align} \label{eq:Rcoop_Rdfct}
    R^{i,\mathrm{coop}}_t(T)
= \sum_{\tau=t}^T \delta^{\tau-t} r^{i,\mathrm{coop}}_\tau,
\hspace{2cm}
R^{i,\mathrm{dfct}}_t(T)
= D^i_t + \sum_{\tau=t+1}^T \delta^{\tau-t} r^{i,\mathrm{NE}}.
\end{align}
By \eqref{eq:sigma-star}, for all $t \in \mathbb{N}$,
\begin{align*}
    \sum_{\tau=t}^{\infty} \delta^{\tau-t} r^{i,\mathrm{coop}}_\tau
\ge
D^i_t + \sum_{\tau=t+1}^{\infty} \delta^{\tau-t} r^{i,\mathrm{NE}} + \varepsilon^i,
\end{align*}
and therefore
\begin{align*}
    R^{i,\mathrm{coop}}_t(T) - R^{i,\mathrm{dfct}}_t(T)
& \ge
\varepsilon^i - \sum_{\tau=T+1}^{\infty} \delta^{\tau-t} \big(r^{i,\mathrm{coop}}_\tau - r^{i,\mathrm{NE}}\big) \\
& \ge \varepsilon^i - \frac{\delta^{T+1-t}}{1-\delta}\, C^i,
\end{align*}
where $C^i$ is as above. Defining
\begin{align}\label{ineq:slack}
\Phi_t^i(T) := R_t^{i,\mathrm{coop}}(T) - \sum_{\tau=t}^T \delta^{\tau-t} d^i - R_t^{i,\mathrm{dfct}}(T),
\end{align}
we obtain that \eqref{eq:IC} is equivalent to $\Phi_t^i(T) + \delta^{T-t}\bar{w}^i \ge 0$. Furthermore,  \eqref{eq:deposit-bound} and the calculations above give us that $\Phi_t^i(T) \ge \varepsilon^i/2 - \delta^{T+1-t} C^i/(1-\delta)$. 

For $t \le T - M^i$ we have $T + 1 - t \ge M^i + 1$, so by the choice of $M^i$,
\begin{align*}
\delta^{\,T+1-t} \;\le\; \delta^{M^i} \;\le\; \frac{(1-\delta)\varepsilon^i}{2C^i},
\end{align*}
so $\Phi^i_t(T) \ge 0$. Since $\bar{w}^i \ge 0$, inequality \eqref{eq:IC} then holds at every such period independently of the terminal reward. At these periods the margin $\varepsilon^i$ alone outweighs both the truncation tail and the bounded deposit cost, so no terminal payoff is required; a terminal payoff can be needed only over the final $M^i$ periods $t \in \{T - M^i + 1, \dots, T\}$. In particular, if $M^i = 0$ then no terminal payoff is needed for player $i$, and we set $\bar{w}^i = 0$. We assume $M^i \ge 1$ in what follows.

For any $t \le T$ the truncated values satisfy $R_t^{i,\mathrm{coop}}(T) \ge -\bar{r}^i/(1-\delta)$ and, since $D^i_t \le \bar r^i$, also $R_t^{i,\mathrm{dfct}}(T) \le \bar{r}^i/(1-\delta)$, so by \eqref{ineq:slack} and \eqref{eq:deposit-bound},
\begin{align*}
\Phi_t^i(T) \;\ge\; -\frac{\bar{r}^i}{1-\delta} - \frac{\varepsilon^i}{2} - \frac{\bar{r}^i}{1-\delta} = -\left( \frac{2\bar{r}^i}{1-\delta} + \frac{\varepsilon^i}{2} \right).
\end{align*}
For $t \in \{T - M^i + 1, \dots, T\}$ we have $T - t \le M^i - 1$, hence $\delta^{T-t} \ge \delta^{M^i-1}$, and so
\begin{align*}
\Phi_t^i(T) + \delta^{T-t}\bar{w}^i \ge \Phi_t^i(T) + \delta^{M^i-1}\bar{w}^i \ge 0,
\end{align*}
which is exactly \eqref{eq:IC}. Together with the earlier periods $t \le T - M^i$, this establishes \eqref{eq:IC} at every period $t \le T$.

It remains to verify that the accumulated deposits can finance $\bar{w}^i$. The constant deposit scheme accumulates $\sum_{t=1}^T d^i = T d^i$, growing linearly with the horizon. In particular, the feasibility condition $\bar{w}^i \le T d^i$ holds for every $T \geq T^\ast$. 

This is the asymmetry at the heart of the construction: the discounted cost of the deposit stream is bounded by $\varepsilon^i/2$ for every $T$, while the undiscounted stake $T d^i$ grows without bound, so for $T$ large enough it funds the constant terminal payoff $\bar{w}^i$. 

For $T \ge T^\ast$, no on-path action or deposit deviation, and no off-path deviation, is profitable, so by the one-shot deviation principle $(\tilde{\sigma}, \tilde{\rho})$ is a subgame perfect equilibrium of $G^T_\delta(d, w)$. By construction $\tilde{\sigma}$ prescribes the action profile of $\mathbf{a}^\ast$ at every on-path period $t \le T$.
\end{proof}

The threshold $T^\ast$ is sufficient for the existence of a feasible deposit scheme, but not necessary, and it is far from tight. It is driven by the ratio $\bar{w}^i / d^i$, which becomes large as $\delta$ approaches one, and the reason is that the constant deposit scheme is the crudest admissible one. Solving the design program of Section~\ref{sec:optimal} instead yields the horizons reported in Figure~\ref{fig:PD}, which are smaller by orders of magnitude, so the gap between $T^\ast$ and the minimal horizon is a property of the existence argument rather than of the mechanism.

We remark that the construction leaves free the disposal of the stake forfeited by the deviator. The only requirement is that the deviator does not recover the stake; it may be divided among the remaining players, or given to the neutral agent operating the deposit mechanism.

Theorem~\ref{thm:main} reproduces the on-path actions of any sequence in $\mathcal{A}^\ast$ as the sequence of play of a subgame perfect equilibrium of a finitely repeated game, and hence also reproduces its entire stage-payoff sequence. The realised payoffs differ from the infinite-horizon benchmark by the discounted cost of the deposits, which the proof bounds by $\varepsilon^i/2$ for each player and which can be driven to zero by lengthening the horizon while shrinking the per-period deposit accordingly. In this sense the deposit mechanism recovers the Nash-threat folk theorem in finite time: every feasible payoff profile that strictly dominates, for each player, that player's worst stage-game Nash payoff, and that can be supported by Nash threats with a strict margin, is attained, up to a vanishing deposit cost, by a subgame perfect equilibrium of a finitely repeated game augmented with deposits.

Our construction can be applied to an arbitrary stage-game with any finite number of players and it requires no transfers between players. Since this is where the commitment requirement of the construction resides, it is worth setting out what the intermediary must be able to do. It must observe the action profile and the deposit decisions, hold the deposits, honour the refund rule $w$ announced before play begins, and remain solvent, the last of which is immediate because the rule never returns to a player more than that player has deposited. The intermediary takes no action of its own, receives no payoff, and under the rule of Theorem~\ref{thm:main} moves no funds between players. Within the game the deposits are voluntary and incentive compatible at every stage, so the mechanism is self-enforcing. Moreover, our proof is constructive, delivering an explicit horizon threshold $T^\ast$, deposit $d^i$, and terminal payoffs $\bar{w}^i$ for any target profile.

The threat used in the proof, permanent reversion to the deviator's worst stage-game Nash equilibrium, is the harshest available among stage-game Nash equilibria and is convenient because it is, by definition, self-enforcing. It is not, however, the only possibility. The deposit mechanism can sustain a punishment continuation in the same way that it sustains the cooperative path, which frees the designer from confining threats to stage-game Nash equilibria. Section~\ref{subsec:public_goods_game} exploits this in a public goods setting, where reverting to universal defection is collectively wasteful and a gentler continuation, itself supported by deposits, sustains cooperation at a lower social cost off the equilibrium path.

The proof establishes existence using the simplest admissible scheme, a constant deposit, but the conditions sufficient for the proof to work are met by a larger family of schemes. Because deposits are costly in present value, it is better to collect them as late as the incentive constraints allow. In the next section, we discuss cost-minimising schemes that sustain this form of cooperation through a problem that is linear in the deposits and the terminal payoffs. Finally, although the theorem is stated for repeated stage-games, its argument rests only on a uniform bound on the one-shot deviation gain along the target path, not on stationarity of the payoffs. Section~\ref{subsec:tragedy_of_the_commons} uses this to carry the mechanism over to a dynamic common-pool resource with an endogenous stock, where the stage-game itself evolves with the state.

\subsection{A Nash-threat folk theorem with deposits}\label{subsec:folk-finite}

Proposition~\ref{prop:folk} and Theorem~\ref{thm:main} combine directly. The first says that every payoff in $F^{\mathrm{NE}}$ is generated by some sequence in $\mathcal{A}^\ast$; the second says that the first $T$ terms of any such sequence are the equilibrium play of the augmented finitely repeated game. What remains is to bound the two quantities that separate the realised finite-horizon payoff from the target: the truncation of the sequence at date $T$, and the discounted cost of the deposits.

\begin{corollary}[Nash-threat folk theorem with deposits]\label{cor:folk}
Let $v \in F^{\mathrm{NE}}$ and let $\xi > 0$. There exists $\underline{\delta} < 1$ such that for every $\delta \in (\underline{\delta}, 1)$ there is a threshold $T^\ast$ with the following property: for every horizon $T \ge T^\ast$ there are a deposit scheme $d$ and a terminal-payoff rule $w$, satisfying the feasibility condition $w^i(h_T) \le \sum_{t=1}^{T} p^i_t d^i_t$, such that $G^T_\delta(d,w)$ admits a subgame perfect equilibrium whose normalised on-path payoff vector lies within $\xi$ of $v$.
\end{corollary}

\begin{proof}
Let $\underline{\delta}$, $\eta$ and $\mathbf{a} \in \mathcal{A}^\ast$ be as in Proposition~\ref{prop:folk}, with associated margins $(\varepsilon^i)_{i \in N}$, and observe that \eqref{eq:sigma-star} continues to hold when each $\varepsilon^i$ is replaced by any smaller positive number.

Consider first the truncation. Writing $v^i = (1-\delta)\sum_{t=1}^{T} \delta^{t-1} r^{i,\mathrm{coop}}_t + \delta^{T} \omega^i_T$, where $\omega^i_T$ denotes the normalised value of the discarded tail, the normalised payoff of the truncated path is $(v^i - \delta^T \omega^i_T)/(1-\delta^T)$ and therefore differs from $v^i$ by at most $\eta\, \delta^{T}/(1-\delta^{T})$, since $\lvert \omega^i_T - v^i \rvert \le \eta$ by Proposition~\ref{prop:folk}. This is where the uniform continuation bound pays for itself: the discarded tail is worth almost exactly what the whole path is worth, so no separate approximation argument is needed for the finite horizon.

Consider next the deposits. The proof of Theorem~\ref{thm:main} bounds their total discounted cost by $\varepsilon^i/2$, hence by $(1-\delta)\varepsilon^i/\big(2(1-\delta^T)\big)$ after normalisation. Replacing each $\varepsilon^i$ by $\min\{ \varepsilon^i,\, \xi/(2(1-\delta)) \}$ makes this at most $\xi/2$ whenever $\delta^T \le 1/2$, and enlarging $T$ further makes the truncation term at most $\xi/2$ as well. Applying Theorem~\ref{thm:main} to $\mathbf{a}$ with the reduced margins yields the threshold $T^\ast$ and the scheme $(d,w)$.
\end{proof}

Shrinking the margin lowers the cost of the mechanism but raises the number $M^i$ of final periods over which the terminal reward has to do the work, and hence raises $T^\ast$. The approximation in Corollary~\ref{cor:folk} is thus a trade-off between the horizon and the cost of the deposits rather than a limitation of the construction: for any fixed horizon the mechanism sustains the target exactly, and it is only the deposit cost, not the play, that separates the realised payoff from $v$.

\section{Optimal deposit design}\label{sec:optimal}

Theorem~\ref{thm:main} establishes existence of a feasible deposit scheme by exhibiting a constant per-period deposit collected from the first period onwards together with a terminal payoff. This is far from the only scheme that sustains a given target strategy. Many deposit streams satisfy the incentive constraints, and they are not equivalent in cost, because a deposit paid early is worth more in present value than the same deposit paid late. Collecting the stake early is therefore wasteful, and it is natural to ask which scheme sustains the target at the least cost.

This section makes three points. First, for a fixed target profile and horizon $T$, finding the cheapest scheme is a linear program, one per player, whose objective reduces to a weighted sum of the deposits, the weight on date $t$ being the loss in present value incurred while a unit deposited at $t$ is held. Second, the terminal reward that any feasible scheme must assemble is bounded below by an expression involving only the incentive gap of the game without deposits, and this bounds below in particular the reward $\bar{w}^i$ used in the proof of Theorem~\ref{thm:main}. Third, the same program can be extended to accommodate further requirements that a designer might have, each entering as a modification of a single constraint. For instance, one may ask that cooperation be preferred to deviation by a fixed margin rather than merely weakly, that no deposit exceed the reward obtained in that period, or that the accumulated deposits also cover a fee for the intermediary. Each of these leaves the linear-program structure intact.

On the equilibrium path every player recovers their deposits upon termination, so the only loss relative to the infinite-horizon benchmark is the difference in value of the deposits between the time they are paid and the termination time in which they are returned, due to discounting. We take this cost as the design objective. Because deposits are self-funded and never transferred between players, the problem separates across players: a player's deposit stream and terminal reward enter only that player's own incentive and feasibility constraints, and the program below is solved independently for each $i \in N$. Throughout this section we fix a target profile $\mathbf{a}^\ast \in \mathcal{A}^\ast$, a finite horizon $T$, and a player $i$. We write $D^i_t := \max\{ r^{i,\mathrm{coop}}_t, r^{i,\mathrm{dfct}}_t \}$ for the value of the most profitable one-shot deviation in period $t$, exactly as in the proof of Theorem~\ref{thm:main}.

The cost-minimising scheme solves
\begin{align}
\min_{(d^i_t)_{t=1}^{T},\, w^i \ge 0} \quad & -\delta^{T-1} w^i + \sum_{t=1}^{T} \delta^{t-1} d^i_t \label{eq:opt} \\
\text{subject to} \quad
& \sum_{\tau=t}^{T} \delta^{\tau-t}\big( r^{i,\mathrm{coop}}_\tau - d^i_\tau \big) + \delta^{T-t} w^i \;\ge\; D^i_t + \sum_{\tau=t+1}^{T} \delta^{\tau-t} r^{i,\mathrm{NE}}, \quad t \in \{1,\dots,T\}, \label{eq:incentive} \\
& \sum_{t=1}^{T} d^i_t \;\ge\; w^i, \label{eq:feasibility} \\
& 0 \le d^i_t \le r_t^{i,\mathrm{coop}}, \quad t \in \{1,\dots,T\}.  \label{eq:liquidity}
\end{align}
The constraint \eqref{eq:incentive} is the incentive condition \eqref{eq:IC} from the proof of Theorem~\ref{thm:main}, written for an arbitrary deposit stream: at every period the discounted continuation value on the path, net of the deposits still to be paid and augmented by the terminal reward, must weakly exceed the value of the best one-shot deviation followed by reversion to the Nash threat. The constraint \eqref{eq:feasibility} is feasibility: the accumulated deposits must cover the terminal reward, so that the mechanism is self-funded and calls on no external resources. The constraint \eqref{eq:liquidity} is a liquidity bound, capping the period-$t$ deposit at $r_t^{i,\mathrm{coop}}$. Two remarks are in order. The constant scheme of Theorem~\ref{thm:main} deposits a positive amount in every period and therefore violates \eqref{eq:liquidity} whenever the target pays player $i$ nothing at some date, as happens in the rotation of Subsection~\ref{subsec:congestion_game}; the theorem corresponds to the program without the liquidity bound. Consequently $T \ge T^\ast$ does not guarantee that the capped program is feasible, and we write $T_{\min}$ for the least horizon at which it is. A necessary condition is $\sum_{t=1}^T r^{i,\mathrm{coop}}_t \ge w^i$, so a target that pays player $i$ zero too often cannot be sustained under the cap at any horizon. All constraints are linear in $(d^i_t)_t$ and $w^i$, so \eqref{eq:opt} is a linear program, and it is small: it has $T+1$ variables and $2T+1$ inequalities.

The optimum of \eqref{eq:opt} has a simple structure, which we record because it makes precise the sense in which the deposits are costly.

\begin{lemma}\label{lem:cost}
Suppose \eqref{eq:opt} is feasible. Then it attains its optimum, at every optimum the terminal reward exhausts the accumulated stake, $w^i = \sum_{t=1}^{T} d^i_t$, and the optimal value equals
\begin{align*}
    \sum_{t=1}^{T} \big( \delta^{t-1} - \delta^{T-1} \big)\, d^i_t .
\end{align*}
\end{lemma}

\begin{proof}
By \eqref{eq:feasibility} and $d^i \ge 0$ the objective is at least $\sum_{t=1}^{T}(\delta^{t-1}-\delta^{T-1})d^i_t \ge 0$, so the program is bounded; being feasible, it attains its optimum. The objective is strictly decreasing in $w^i$, and raising $w^i$ relaxes every incentive constraint \eqref{eq:incentive}, in which $w^i$ appears with the non-negative coefficient $\delta^{T-t}$ on the left-hand side. At an optimum $w^i$ therefore takes the largest value admitted by \eqref{eq:feasibility}, namely $\sum_{t=1}^T d^i_t$. Substituting into the objective gives $\sum_{t=1}^{T} \delta^{t-1} d^i_t - \delta^{T-1}\sum_{t=1}^{T} d^i_t$, which is the stated expression.
\end{proof}

Substituting $w^i = \sum_t d^i_t$ therefore reduces \eqref{eq:opt} to an equivalent program in the deposits alone, with $T$ variables, $2T$ constraints, and the explicit objective above. That objective has a direct reading: a unit deposited at date $t$ is charged $\delta^{t-1} - \delta^{T-1}$, the loss in present value incurred over the interval during which that unit is held by the intermediary, and no other cost is borne. This price is strictly decreasing in $t$ and vanishes at $t = T$. Along the locus $w^i = \sum_t d^i_t$ identified by the lemma, moreover, a change in $d^i_T$ leaves both the objective and every incentive constraint unchanged, since $d^i_T$ enters \eqref{eq:incentive} with coefficient $-\delta^{T-t}$ and $w^i$ with coefficient $\delta^{T-t}$. Whenever $r^{i,\mathrm{coop}}_T > 0$ the cost-minimising scheme is therefore not unique, the optimal set containing a segment along which $d^i_T$ and $w^i$ move together.

The terminal reward is also bounded below by a quantity that depends on the target profile and the horizon alone, and not on the deposit stream chosen to fund it. Write
\begin{align*}
    G^i_t(T) := R^{i,\mathrm{coop}}_t(T) - R^{i,\mathrm{dfct}}_t(T)
\end{align*}
for the incentive gap at date $t$ of the game without deposits, that is, for the quantity $\Phi^i_t(T)$ of \eqref{ineq:slack} evaluated at a zero deposit stream, so that a negative gap identifies a date at which the target fails to be self-enforcing on its own.

\begin{proposition}\label{prop:wreq}
Let $d^i \ge 0$ and $w^i \ge 0$ satisfy the incentive constraints \eqref{eq:incentive}. Then
\begin{align*}
    w^i \;\ge\; w^{i,\mathrm{req}}(T) := \max_{t \le T} \ \delta^{\,t-T} \big( -G^i_t(T) \big)^{+} .
\end{align*}
\end{proposition}

\begin{proof}
Rearranging \eqref{eq:incentive} gives $\sum_{\tau=t}^{T} \delta^{\tau-t} d^i_\tau - \delta^{T-t} w^i \le G^i_t(T)$ for every $t \le T$. Deposits are non-negative, so the sum on the left is non-negative, whence $\delta^{T-t} w^i \ge -G^i_t(T)$. As $w^i \ge 0$, the bound follows on taking the maximum over $t$.
\end{proof}

The hypothesis is weaker than feasibility for \eqref{eq:opt}, since it invokes neither the liquidity cap nor \eqref{eq:feasibility}, so the bound applies both to every scheme admissible for the program and to the scheme constructed in the proof of Theorem~\ref{thm:main}. In particular $w^{i,\mathrm{req}}(T) \le \bar{w}^i$, so the reward chosen there for convenience is never smaller than the one the target requires.

The program as stated treats the intermediary as costless, allows each agent to deposit up to the on-path payoff $r_t^{i,\mathrm{coop}}$, and asks only that cooperation be weakly preferred to deviation. Each of these idealisations can be dropped, and dropping any one of them changes a single constraint, leaving the linear-program structure intact.

The incentive constraints in \eqref{eq:opt} hold with equality at the optimum, so cooperation is only just preferred to deviation and is fragile to misspecification. Replacing the right-hand side of the incentive family by $D^i_t + \sum_{\tau=t+1}^{T} \delta^{\tau-t} r^{i,\mathrm{NE}} + \eta$, for a margin $\eta \ge 0$, requires on-path play to dominate deviation by at least $\eta$ in every period. The equilibrium then tolerates payoff misspecification, occasional non-best-response play by opponents, or error in evaluating continuation values, up to $\eta$. The optimal cost is non-decreasing in $\eta$, which is the price of robustness, and because a larger early stake strengthens the late-period incentives, robustness pushes deposits earlier, working against the deferral that pure cost minimisation produces. In an international agreement, where early withdrawal is the main hazard, front-loading the stake in this way is what makes commitment credible.

The neutral intermediary need not operate for free. Replacing the feasibility constraint by $\sum_{t=1}^T d^i_t \ge w^i + \kappa$, for a fee $\kappa \ge 0$, lets the intermediary retain $\kappa$ on the path while the player still recovers $w^i$. Off the equilibrium path, Theorem~\ref{thm:main} places no restriction on the disposition of deposits forfeited after a deviation, so these too may accrue to the intermediary at no cost to the on-path incentives. The mechanism is therefore operable by a self-funding neutral party and the fee enters the design only through the larger stake it forces the player to deposit, and hence through a higher cost.

Finally, the liquidity bound itself can be varied. Capping the period-$t$ deposit at the current payoff is natural in some settings but not all: in the public goods game of Subsection~\ref{subsec:public_goods_game}, for instance, parties post collateral that exceeds a single period's payoff. Loosening the cap lets the stake be assembled later and more cheaply, lowering both the cost and the minimal horizon $T_{\min}$ at which \eqref{eq:opt} is feasible; tightening it forces deposits into earlier periods, raising both.

The examples below demonstrate the mechanism in different environments and illustrate several of these variations.

\section{Examples}\label{sec:examples}

We illustrate the mechanism through four examples. The first, a repeated prisoner's dilemma (Subsection~\ref{subsec:prisoner_dilemma}), is the canonical environment with a unique, inefficient, stage-game equilibrium. It serves as the simplest test of the mechanism and lets us quantify how little is required to restore cooperation: across a wide range of payoffs the minimum horizon is as small as $T_{\min}=2$. The second, a congestion game (Subsection~\ref{subsec:congestion_game}), moves from symmetric cooperation to a strategy in which players take turns. The third, an $n$-player public goods game (Subsection~\ref{subsec:public_goods_game}), shows that reverting to the harshest stage-game equilibrium is unnecessary. It shows that the mechanism can sustain a gentler, non-equilibrium punishment as a continuation, trading deterrence against off-path welfare, and it dispenses with the liquidity cap. The final example, a dynamic common-pool resource with an endogenous stock (Subsection~\ref{subsec:tragedy_of_the_commons}), takes the mechanism beyond repeated stage-games, showing that the construction needs only a uniform bound on the one-shot deviation gain, not a stationary payoff structure.

\subsection{Repeated prisoner's dilemma}\label{subsec:prisoner_dilemma}

Consider the infinitely repeated prisoner's dilemma with stage payoffs normalised so that mutual defection yields $0$, unilateral defection yields $r^{\mathrm{dfct}} = 1$ and mutual cooperation yields $r^{\mathrm{coop}} \in (0,1)$. Each repeated prisoner's dilemma then corresponds to a single point $(\delta, r^{\mathrm{coop}}) \in (0,1)^2$. For each such point we solve the program \eqref{eq:opt}, and we record the minimal implementation horizon $T_{\min}$ that sustains mutual cooperation at every period. Figure~\ref{fig:PD} reports $T_{\min}$ over $(0,1)^2$. The solid line $r^{\mathrm{coop}} = 1 - \delta$ is the infinite-horizon incentive boundary: in the hatched region below it, mutual cooperation cannot be sustained even in the infinitely repeated game, so condition \eqref{eq:sigma-star} does not hold. Above the boundary the deposit mechanism restores cooperation, with $T_{\min}$ increasing as the boundary is approached and falling to $T_{\min} = 2$ away from it. These horizons are orders of magnitude below the sufficient threshold $T^\ast$ of Theorem~\ref{thm:main}, which illustrates how much is lost by fixing the deposit to be constant. To locate empirically relevant cases we overlay canonical parametrisations from the literature, rescaled to the normalised payoffs: \citet{axelrod1981evolution} specifies no discount factor and appears as the horizontal line $r^{\mathrm{coop}} = 0.5$, while \citet{cason2019individual, romero2018constructing, dal2011evolution} report a discount (or continuation) factor and appear as points. For most of these games a horizon of two or three periods already suffices, so that very limited commitment is enough to sustain cooperation.

Because the target is stationary, the quantities entering Proposition~\ref{prop:wreq} can be computed in closed form, and doing so accounts for the features of Figure~\ref{fig:PD}. Write $k := T - t$ for the number of periods that remain, and note that $D_t = 1$ at every date because $r^{\mathrm{coop}} < 1$. The incentive gap is then
\begin{align*}
    G_t(T) = r^{\mathrm{coop}}\,\frac{1 - \delta^{\,k+1}}{1-\delta} - 1 ,
\end{align*}
which is increasing in $k$, so that $G_t(T) \ge 0$ if and only if $\delta^{\,k+1} \le 1 - (1-\delta)/r^{\mathrm{coop}}$. Since $\delta^{\,k+1} \to 0$, such a $k$ exists precisely when the right-hand side is positive, that is when $r^{\mathrm{coop}} > 1 - \delta$, which is the boundary drawn in Figure~\ref{fig:PD}; below it the gap is negative at every date and no horizon can help. Above it the gap is negative precisely at the dates with fewer than
\begin{align*}
    k^\ast = \left\lceil \frac{\ln\!\big( 1 - (1-\delta)/r^{\mathrm{coop}} \big)}{\ln \delta} \right\rceil - 1
\end{align*}
periods remaining, a count that plays here the role of the constant $M^i$, which the proof of Theorem~\ref{thm:main} bounds rather than computes. On that range
\begin{align*}
    \delta^{-k}\big(-G_t(T)\big) = \delta^{-k}\Big( 1 - \frac{r^{\mathrm{coop}}}{1-\delta} \Big) + \frac{\delta\, r^{\mathrm{coop}}}{1-\delta} ,
\end{align*}
whose first coefficient is negative because $r^{\mathrm{coop}} > 1-\delta$, so the expression is decreasing in $k$, the maximum in Proposition~\ref{prop:wreq} is attained at $k = 0$, and
\begin{align*}
    w^{\mathrm{req}} = 1 - r^{\mathrm{coop}} = r^{\mathrm{dfct}} - r^{\mathrm{coop}} .
\end{align*}
No scheme can sustain mutual cooperation with a stake smaller than the one-shot temptation faced in the final period, and this floor depends neither on the discount factor nor on the horizon. It is the last period's incentive constraint that binds in Proposition~\ref{prop:wreq}, even though the gap is negative at every date with fewer than $k^\ast$ periods remaining.

\begin{figure}[!ht]
  \centering
  \includegraphics[width=0.65\textwidth]{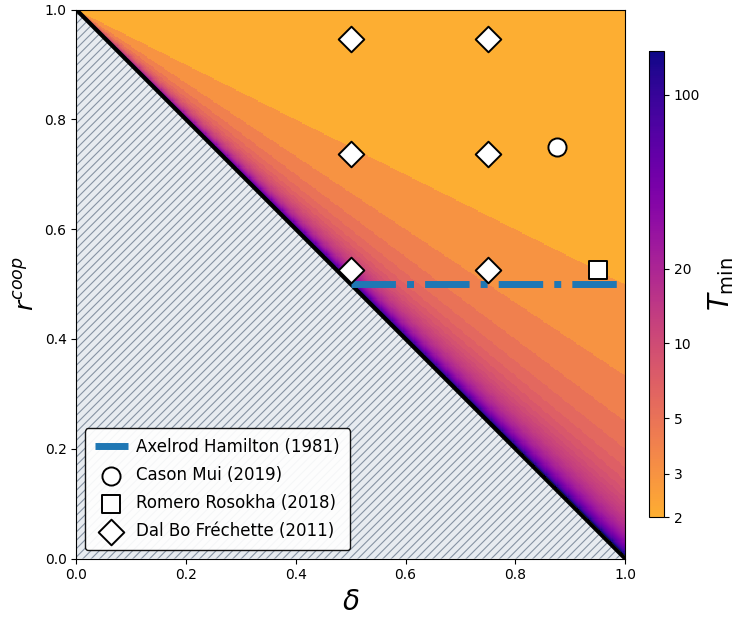}
  \caption{Minimal horizon sustaining cooperation in the repeated prisoner's
  dilemma. With payoffs normalised to $r^{\mathrm{NE}} = 0$ and
  $r^{\mathrm{dfct}} = 1$, every repeated prisoner's dilemma is a point
  $(\delta, r^{\mathrm{coop}}) \in (0,1)^2$. Colour gives the smallest horizon
  $T_{\min}$ that allows cooperation to be sustained. In the hatched region
  below the black line $r^{\mathrm{coop}} = 1 - \delta$ cooperation fails even
  with an infinite horizon. Overlaid are canonical prisoner's dilemma games
  rescaled to these payoffs.}
  \label{fig:PD}
\end{figure}

\subsection{Congestion game}\label{subsec:congestion_game}

Consider a finite-horizon game with $n \ge 3$ players in which each player chooses between actions $\rA{in}$ and $\rA{out}$. Action $\rA{out}$ yields a payoff of zero. If $m \ge 1$ players choose $\rA{in}$, each of them receives $\phi(m)$, where $\phi$ is strictly decreasing with $\phi(1) > 0$ and $\phi(n) = 0$, so the payoff from the resource is highest when a single player uses it and vanishes when all do. This is a congestion game, in which the payoff from a resource falls with its congestion level \citep{christodoulou2005price}. We assume the map $m \mapsto m\,\phi(m)$ is decreasing, so aggregate welfare is greatest when exactly one player chooses $\rA{in}$. We focus our attention on a rotating allocation in which the players take turns playing $\rA{in}$, so that exactly one player receives $\phi(1)$ in each period and the others receive zero. This allocation is efficient (it obtains the maximal aggregate welfare) but not self-enforcing: a player whose turn it is not can deviate to $\rA{in}$, raising congestion and earning $\phi(2) > 0$ rather than zero, and backward induction makes this deviation attractive in the closing periods.

We design a deposit scheme that sustains this rotation for a three-player example. Figure~\ref{fig:coordination} plots three quantities per player. The cooperation and deviation values $R^{i,\mathrm{coop}}_t(T)$ and $R^{i,\mathrm{dfct}}_t(T)$ are those of \eqref{eq:Rcoop_Rdfct}. The deposit scheme enters through the net deposit value
\begin{align*}
    V^i_t(d^i,\bar{w}^i) := -\sum_{\tau = t}^T \delta^{\tau-t}d^i_\tau + \delta^{T-t} \bar{w}^i,
\end{align*}
which is the discounted terminal reward net of the deposits still to be paid from period $t$ onwards. Its first period value $V^i_0(d^i,\bar{w}^i)$ measures the net present cost of the mechanism to player $i$. The cost is minimal and may be unequal across players: player 1 cost is zero as they need no terminal reward, while players 2 and 3 carry only a negligible cost. Adding $V^i_t(d^i,\bar{w}^i)$ to $R^{i,\mathrm{coop}}_t(T)$ provides the on-path continuation value under the mechanism. 

\begin{figure}[!ht]
  \centering
  \includegraphics[width=0.65\textwidth]{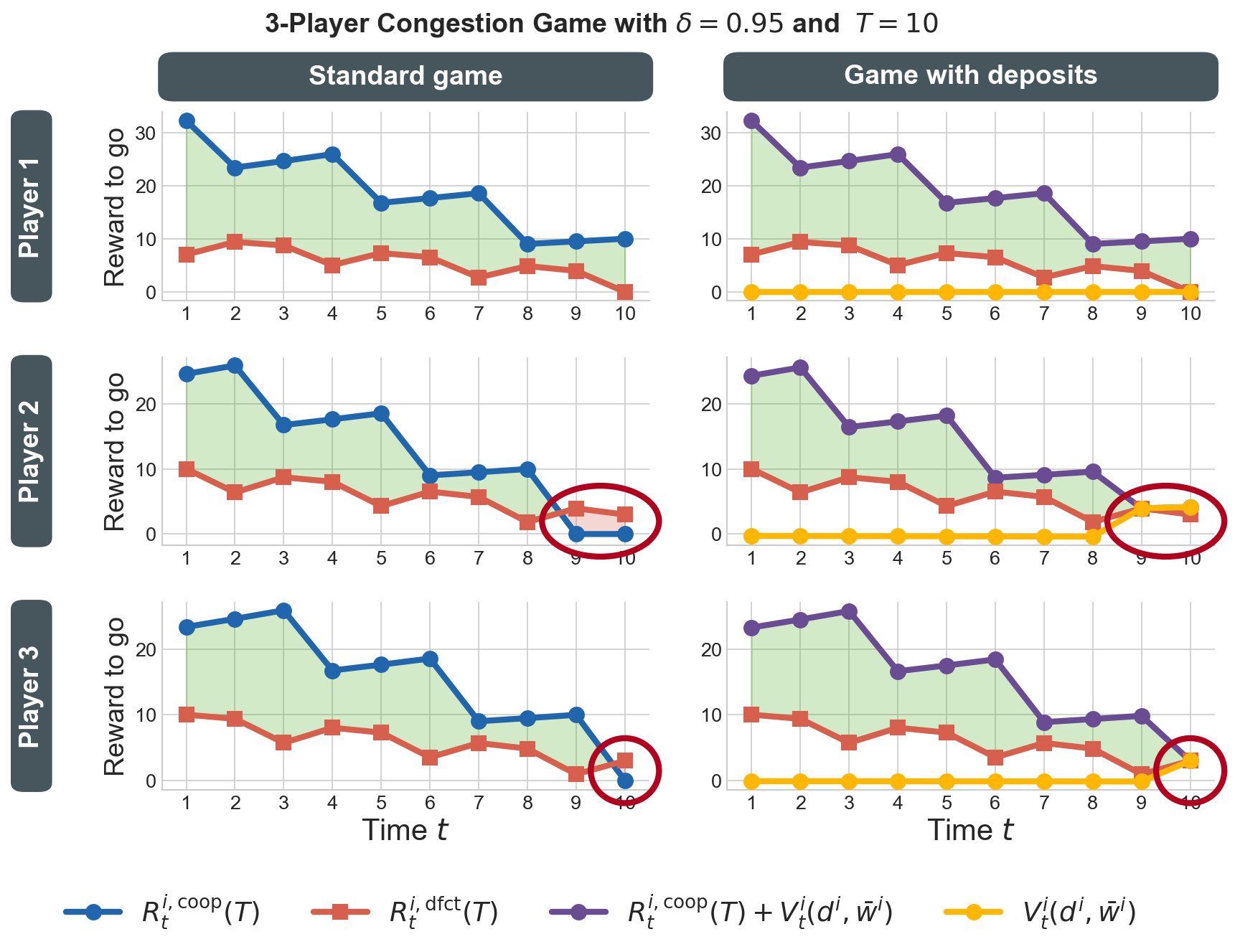}
  \caption{Deposits sustaining a rotating allocation in the congestion game.
  Players take turns playing action $\rA{in}$. This strategy is efficient but
  not self-enforcing. Without deposits (left), the deviation value
  $R^{i,\mathrm{dfct}}_t$ (red) overtakes the cooperation value
  $R^{i,\mathrm{coop}}_t$ (blue) near the end of the game for players 2 and 3.
  With deposits and a terminal reward (right), the cooperation value
  $R^{i,\mathrm{coop}}_t - V^i_t$ (purple) stays above
  $R^{i,\mathrm{dfct}}_t$ throughout. The cost of depositing $V^i_t$ is
  negligible for all players and zero for player 1.}
  \label{fig:coordination}
\end{figure}
The mechanism's demands scale predictably with the number of players $n$. Under the rotation, each player enjoys $\phi(1)$ every $n$ periods, while the temptation to defect, worth $\phi(2)$, recurs in all the remaining periods. As $n$ grows, the cooperative reward becomes sparser while the temptation persists, so the terminal stake needed to deter deviation grows with $n$.

\subsection{Public goods game}\label{subsec:public_goods_game}

Consider a public goods game with $n$ players. In each period every player chooses to cooperate, contributing one unit to a public pot, or to defect, contributing nothing. The total contribution is multiplied by a synergy factor $\lambda \in (1,n)$ and shared equally among all $n$ players, irrespective of their own contributions. Such games are a standard model for climate dilemmas \citep{hintze2020inclusive, hasson2010climate, luo2025cooperative}. If $n_C$ players cooperate, each of them receives $\lambda n_C / n - 1$, while each of the remaining players receives $\lambda n_C / n$. A player who switches unilaterally from cooperation to defection lowers $n_C$ by one and so raises their own stage payoff by $1 - \lambda/n$, which is positive because $\lambda < n$. Defection is therefore dominant, the unique stage-game Nash equilibrium has every player defecting for a payoff of zero, and yet full cooperation would give each player $\lambda - 1 > 0$. This is the tragedy of the commons, which generalises the prisoner's dilemma: cooperation is efficient but individually unstable.

The deposit mechanism allows cooperation to be sustained as a subgame perfect equilibrium in the finitely repeated game, provided the horizon is long enough for the accumulated stakes to discipline behaviour in the final periods. But the construction punishes any deviation by reverting to the stage-game's unique Nash equilibrium, an outcome that is bad for every player, not just the deviator, and, in the climate context, bad for the environment itself, since it prescribes a collapse to zero abatement.

\begin{figure}[!ht]
  \centering
  \includegraphics[width=0.65\textwidth]{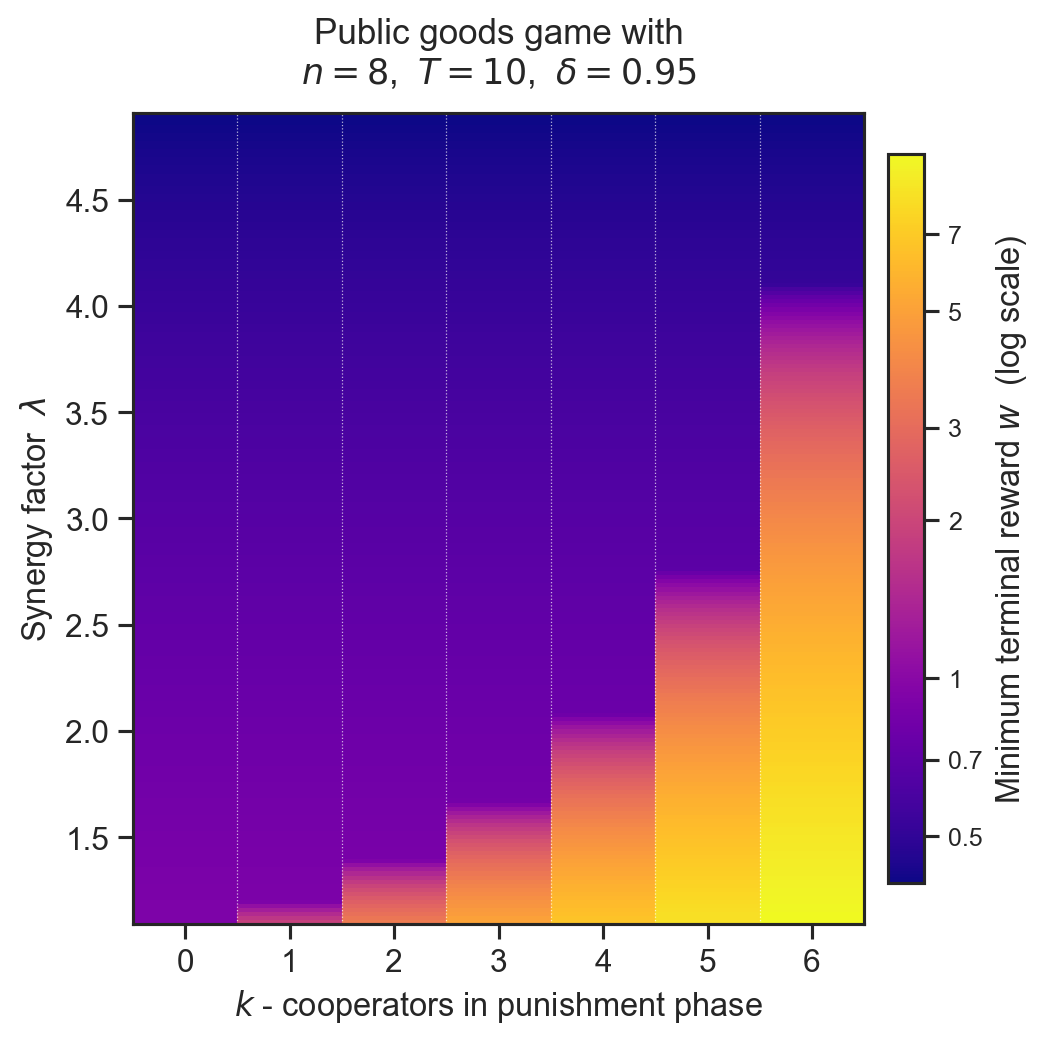}
  \caption{Cost of leniency in the public goods game. The heatmap reports the
  minimum terminal reward $w$ sustaining full cooperation in a public goods
  game ($n = 8$, $T = 10$, $\delta = 0.95$), as a function of the synergy
  factor $\lambda$ and the number of cooperators $k$ retained after a
  deviation. Following a defection the deviator defects forever while the
  remaining players run a $k$-cooperator strategy, itself sustained by
  deposits; $k = 0$ is reversion to universal defection, the worst stage-game
  Nash equilibrium, as in Theorem~\ref{thm:main}. Since the deviator
  free-rides at $\lambda k/n$, more lenient punishments (larger $k$) demand a
  larger reward, while higher synergy $\lambda$ lowers it.}
  \label{fig:pgg}
\end{figure}

We therefore consider now a gentler continuation, indexed by a punishment level $k$. After a deviation, we assume that the deviator will defect in every subsequent period, while the remaining $n-1$ players switch to a strategy in which $k$ of them cooperate and the other $n-1-k$ defect at each period. The parameter $k$ interpolates between the harshest threat and more lenient ones. The $k$-cooperator strategy is not a stage equilibrium for $k \geq 1$. However, the key observation is that the deposit mechanism can sustain this punishment continuation just as it sustains the cooperative path, so the designer is not restricted to stage-game Nash equilibria when choosing a threat. Leniency is collectively desirable: reverting to everyone defecting leaves every player with zero, whereas retaining $k$ players cooperating leads to an aggregate welfare of $k(\lambda-1)$. It comes at a price, since it weakens deterrence: the deviator keeps free-riding and so earns $\lambda k/n$ per period, which grows with $k$, making defection more tempting and raising the terminal reward needed to deter it.

Figure~\ref{fig:pgg} makes this trade-off explicit. For each synergy factor $\lambda$ and punishment level $k$, the heatmap reports the minimum terminal reward $w$ that solves \eqref{eq:opt} with $\lambda k/n$ instead of $r^{\mathrm{NE}}$ and thus sustains full cooperation as a subgame perfect equilibrium. As anticipated, $w$ rises with $k$, since more lenient punishments demand larger stakes, and falls with $\lambda$, since a higher synergy factor both lowers the one-shot temptation $1-\lambda/n$ and raises the per-period value of cooperation $\lambda-1$. The figure then shows that full cooperation can be sustained while preserving substantial cooperation off the equilibrium path, at the cost of a larger terminal stake. In this way the group never has to threaten the mutually destructive collapse triggered by a single defection.

\subsection{Tragedy of the commons}\label{subsec:tragedy_of_the_commons}

We illustrate the deposit mechanism in a setting where the stage-game is not fixed but evolves endogenously with a shared state variable, showing that the construction turns on a uniform bound on the one-shot deviation gain rather than on stationarity of the payoffs. The following example is a discrete-time version of the Tragedy of the Commons game from~\citet{polasky2006cooperation}. 

The players $N$, with $n \geq 2$, harvest from a common renewable resource over a finite horizon $t = 1, \ldots, T$. The resource stock at the beginning of period $t$ is $S_t \in [0, K]$, where $K > 0$ is the carrying capacity. At each period, each player $i$ simultaneously chooses a harvest level $h^i_{t} \geq 0$. Total harvest is $H_t = \sum_{i \in N} h^i_{t}$, subject to $H_t \leq S_t$, and the stock evolves according to
\begin{align*}
S_{t+1} = f(S_t - H_t),
\end{align*}
where $f \colon [0, K] \to [0, K]$ satisfies $f(0) = 0$, $f(K) = K$, $f'(x) > 0$, $f''(x) < 0$ for all $x \in (0, K)$, and $f'(0) > 1$. Player $i$'s per-period payoff is
\begin{align*}
\pi_i(h^i_{t}, S_t) = [P - c(S_t)] \, h^i_{t},
\end{align*}
where $P > 0$ is the harvest price and $c \colon (0, K] \to \mathbb{R}$ is the unit cost, satisfying $c(S) > 0$, $c'(S) < 0$, $c''(S) > 0$, and $\lim_{S \to 0} c(S) = \infty$. We assume $P > c(K)$, so that harvesting is profitable at high stock levels. The break-even stock $\underline{S} \in (0, K)$ is defined by $P = c(\underline{S})$.

This model captures the canonical tragedy of the commons: each player's extraction imposes a negative externality on all others by depleting the stock and raising future costs. In the infinite-horizon version with discount factor $\delta \in (0,1)$, \citet{polasky2006cooperation} show that a cooperative outcome, sustaining the stock at an optimal steady-state level $S^* \in (\underline{S}, K)$, can be supported as a subgame perfect equilibrium via a two-phase punishment scheme, provided $\delta$ is sufficiently large and $n$ sufficiently small relative to the growth rate. The punishment forces a deviator to extract at a stock level below $\underline{S}$ where rents are negative, calibrated so that the deviator's continuation payoff is exactly zero. However, this construction relies on the infinite horizon: with a finite endpoint, backward induction unravels cooperation from the terminal period.

Although Theorem~\ref{thm:main} applies to repeated stage-games with a fixed payoff structure, the key ideas of the construction carry over to this dynamic setting. The essential observation is that the cooperative path eventually reaches a neighbourhood of the steady state $S^*$, after which the per-period payoffs and deviation incentives stabilise. What matters for the deposit mechanism is not stationarity of the payoff structure, but a uniform bound on the deviation gain along the target path. Since the per-period payoff $[P - c(S_t)] h^i_{t}$ is linear in $h^i_{t}$, the optimal one-shot deviation at any stock level $S_t$ along the cooperative path is to extract as much as the feasibility constraint allows. The resulting deviation gain is increasing in $S_t$ and is uniformly bounded above by some $\bar{\Delta}(S^*)$. The argument from Theorem~\ref{thm:main} then applies with only notational changes: one replaces the fixed cooperation and defection payoffs with their state-dependent counterparts, and the tail truncation error is controlled by the same bound $\max_t |r_t^{i,\mathrm{coop}} - r^{i,\mathrm{NE}}_t| \le C^i < \infty$.

The deposit mechanism then applies as before. Each player $i$ pays a constant deposit $d^i > 0$ at each period along the cooperative path. After any deviation, play reverts to the Nash equilibrium of the dynamic game, in which players extract without regard for future stock levels, and all accumulated deposits are forfeited. The discounted cost of the deposit stream is bounded by $d^i/(1 - \delta)$ regardless of $T$, as shown in \eqref{eq:deposit-bound}, while the undiscounted accumulation $d^i T$ grows linearly with the horizon. As in Theorem \ref{thm:main}, this asymmetry ensures that deposits alone close the incentive gap for all but the final $M^i$ periods, where $M^i$ is a constant independent of $T$. In those final periods, a terminal reward $w^i$, financed by the accumulated deposits, restores the incentive to cooperate. For $T$ sufficiently large, the feasibility constraint $d^i T \geq w^i$ is satisfied, and deviation is unprofitable at every point along the cooperative path.

This suggests that the deposit mechanism can sustain cooperation in dynamic games with endogenous state variables, not only in repeated stage-games.The essential requirement is not stationarity of the payoff structure but a uniform bound on the one-shot deviation gain along the target equilibrium path.

Taken together, the four examples demonstrate the reach of the mechanism. The prisoner's dilemma shows that recovering cooperation can be cheap and quick; the congestion game, that the target may be asymmetric and the cost unevenly shared, even zero for some players; the public goods game, that the off-path threat can be made efficient rather than ruinous; and the common-pool resource, that the construction survives an endogenous state. Each is an instance of the same design, distinguished only by the target profile, the threat, and the liquidity cap of Section~\ref{sec:optimal}.

\section{Conclusion}\label{sec:conclusion}

We have introduced a self-enforcing deposit mechanism that recovers the Nash-threat folk theorem in finite time. Our main result shows that any sequence of pure action profiles which is sustainable by Nash threats with a strict margin in the infinitely repeated game, that is any profile in the class $\mathcal{A}^\ast$, can be reproduced (at least in its initial part) as a subgame perfect equilibrium of the same stage-game repeated over a sufficiently long but finite horizon, once the play is augmented with voluntary deposits and a terminal payoff. The realised payoffs match the infinite-horizon benchmark up to the discounted cost of the deposits, which can be driven to zero by lengthening the horizon while shrinking the per-period deposit accordingly. 

The lesson we draw is conceptual rather than merely technical: the inefficiency of finitely repeated games reflects the absence of instruments for transferring incentives across periods, not the absence of future interaction as such. Once players can post a refundable, self-funded stake that is forfeited after a deviation, an accumulated terminal reward plays exactly the role that a diminished continuation value plays in the infinite-horizon game, and the cooperation that backward induction would otherwise unravel is restored.

The construction is driven by an asymmetry. Because deposits are paid before the terminal date and earn no interest, they are costly in present value; yet their discounted cost remains uniformly bounded while their undiscounted accumulation grows linearly with the horizon. A stream of small, incentive-compatible payments therefore assembles a terminal stake large enough to discipline behaviour in the closing periods, at a present-value cost that vanishes as the horizon grows. The mechanism leaves the stage-game untouched, requires no transfers between players, and presumes only a neutral intermediary that holds the deposits and returns them according to a rule fixed in advance. The deposits are incentive compatible at every stage, so participation by the players never has to be assumed, and in this sense the mechanism is self-enforcing among them; what it does require is that the intermediary can commit to the refund rule, so the construction relocates the commitment requirement rather than removing it.

Our analysis is constructive. For any target profile we exhibit an explicit horizon threshold, per-period deposit, and terminal reward, and we characterise the cost-minimising scheme as a small linear program, solved separately for each player because the problem does not couple them. Discounting makes it optimal to collect deposits as late as the incentive constraints permit, and the same program accommodates three economically natural features of the environment, namely the robustness of cooperation to misspecification, the funding of the intermediary, and the agents' liquidity, with each entering as a modification of a single constraint. This exposes a clean tension between cost and credibility: pure cost minimisation defers the stake, whereas a robustness margin and a tight liquidity cap push it forward, and that front-loading is precisely what makes commitment credible when early withdrawal is the principal hazard.

The reach of the design is wide. It applies to arbitrary stage-games with any finite number of players and accommodates asymmetric targets whose cost falls unevenly and may even be zero for some players. The mechanism allows the off-path threat to be a gentler continuation that is itself sustained by deposits rather than a costly collapse to the stage-game Nash equilibrium, and it extends beyond repeated games to dynamic interactions with an endogenous state, where only a uniform bound on the one-shot deviation gain, and not stationarity of the payoffs, is required. These properties make the mechanism a natural candidate wherever interaction is finite, direct transfers are undesirable or infeasible, and a minimal enforcement technology is available. International agreements on climate or trade, in which parties could pledge refundable collateral such as sovereign bonds without committing to transfers between signatories, are a leading example, and the front-loading induced by a robustness margin speaks directly to the withdrawal risk that such agreements face. Digital environments offer another, since a smart contract can serve as the neutral intermediary that holds and conditionally releases funds, including in the emerging economies of interacting artificial agents that bargain repeatedly but without long-term commitment.

Several assumptions invite relaxation. Our environment has observed actions and deposit decisions, and complete information about the stage-game; extending the mechanism to imperfect public or private monitoring, where deviations are detected only through noisy signals and refunds must be conditioned accordingly, is the most immediate direction. The intermediary, though it maximises nothing, must be able to commit to the refund rule and to remain solvent, so understanding the mechanism's robustness when this party may fail, and how far decentralised enforcement through smart contracts can substitute for it, is a question of practical importance. The common-pool example suggests that the construction admits a general formulation for stochastic and Markov games, in which the uniform deviation-gain bound replaces stationarity, and we conjecture an analogous result at that level of generality. Finally, heterogeneous discounting, asymmetric liquidity across players, and the experimental behaviour of agents facing such a mechanism are all natural next steps. We leave these questions to future work.

% ------------------------------------------------------------
% EJOR / Elsevier required end-of-manuscript statements.
% >>> EDIT THE PLACEHOLDER TEXT BELOW BEFORE SUBMITTING <<<
% Delete the generative-AI declaration entirely if there is
% nothing to disclose.
% ------------------------------------------------------------

\section*{Funding}

Giulio Salizzoni gratefully acknowledges the support of the Swiss National Science Foundation, grant number 207984 and of the SNSF NCCR Automation Grant.

\section*{Declaration of generative AI usage}

During the preparation of this work the authors used Claude in
order to correct grammatical mistakes and typos, and to smooth out a final version of the already human-written paper. After using this tool, the authors reviewed and edited
the content as needed and take full responsibility for the content of the
published article.

\bibliographystyle{apalike}
\bibliography{sn-bibliography}

\end{document}